\documentclass[conference]{IEEEtran}
\IEEEoverridecommandlockouts

\usepackage[letterpaper, top=0.99in, left=0.7in, right=0.7in, bottom=0.78in]{geometry}

\usepackage{cite}
\usepackage{algorithmic}
\usepackage{graphicx}
\usepackage{textcomp}
\usepackage{xcolor}
\def\BibTeX{{\rm B\kern-.05em{\sc i\kern-.025em b}\kern-.08em
    T\kern-.1667em\lower.7ex\hbox{E}\kern-.125emX}}

\usepackage{hyperref}
\hypersetup{pdftex,colorlinks=true,allcolors=blue} 

\usepackage{amsmath, amsthm, amssymb}
\usepackage{verbatim}
\usepackage{graphicx}
\usepackage{array}
\usepackage{mathrsfs}
\usepackage{comment}
\usepackage{enumitem} 
\usepackage{lmodern}

\newtheorem{theorem}{Theorem}
\newtheorem{lemma}{Lemma}

\newcommand{\R}{\mathbb{R}}
\newcommand{\E}{\mathbb{E}}

\def\sA{{\mathsf A}}

\def\sS{{\mathsf S}}

\def\sW{{\mathsf W}}

\def\rd{{\rm d}}
\def\PP{{\mathbb P}}

\def\deq{\triangleq}

\begin{document}

\title{Minimum Distortion Quantization with \\ Specified Output Distribution
\thanks{
 xuaolin@gmail.com}
}

\author{Aolin Xu}
\maketitle

\begin{abstract}
We derive the optimal quantizer of a real-valued random variable $W$ with distribution $P_W$ such that 1) the distribution of the quantization output $X$ that can take $k$ values follows any specified distribution $P_X$ over $\{1,\ldots,k\}$, and 2) the minimum mean squared error (MMSE) of estimating $W$ from $X$ is minimized.
It is shown that the optimal quantizer takes the form $X=\sigma\big(F_{\sigma^{-1}(X)}^{-1}(F_W(W))\big)$, where $\sigma$ is the optimal permutation of $\{1,\ldots,k\}$ among all permutations to minimize the MMSE, and $F$ is the cumulative distribution function. When $P_W$ is uniform over an interval or $P_X$ is uniform over $\{1,\ldots,k\}$, the quantizer takes a simple form $X=F_{X}^{-1}(F_W(W))$.
The concept of majorization plays a key role in the optimality proof. 
Specifying the output distribution is useful for designing quantizers with explicitly controlled output entropy, maximized mutual information between input and output, tailored output distribution to match channel input requirements for communication, and data anonymization.
\end{abstract}

\section{Introduction}

Quantization is a fundamental operation in signal processing, digital communications, data analysis, and artificial intelligence: a continuous-valued source is mapped to a finite alphabet for storage, transmission, or downstream inference. Classical quantizer design focuses on selecting quantization regions and reconstruction points to minimize a distortion metric, most commonly the mean squared error, leading to the well known iterative Lloyd-Max algorithm \cite{Lloyd1982, Max1960} and its variants \cite{quant_1998} with local optimality guarantee. In many applications, however, distortion is not the sole design criterion. The distribution of the quantization indices can be as important as the distortion itself because it directly determines 1) the compressibility of the output index stream \cite{Gallager_2008}, or on the other hand, its informativeness about the source, 2) the feasibility or the effectiveness of transmitting the indices over channels with input constraints \cite{MarcusRothSiegel}, and 3) the statistical properties of the clustered data in data analysis or those of the released data in privacy-aware scenarios \cite{REBOLLOMONEDERO2013185, dist_prsv_k_anonym}.

Quantizer design beyond solely minimizing distortion has been studied over decades. For example, low entropy quantization \cite[Chap. 3]{Gallager_2008} aims to constrain the output entropy; on the other hand, the output entropy can also be maximized \cite{quant_max_ent}; input distribution preserving through dithering has been proposed in \cite{dpq2010}; information maximizing quantizers have been proposed for communication links \cite{boa12, Zeitler12} and for estimating correlated variables \cite{info_distill_quantiz}.
Specifying quantization output distribution is also connected to optimal transport, e.g.\ optimal transport may be used to show the existence of randomized quantizer with specified output distribution \cite{random_OT_quant}, and semi-discrete optimal transport can be used to achieve a specified output distribution when the reconstruction points are prespecified and not optimized, which leads to the Laguerre or power cell type of solutions \cite{power_cell_cluster, var_disc_ot}. In recent years, huge amount of effort has also been spent on designing quantizers to slim down large neural networks and language models while maintaining their inference performance \cite{qnn_survey, eval_q_LLM}. 

This paper studies deterministic scalar quantization with an explicit output distribution requirement. Given a real-valued random variable $W$ with known distribution $P_W$ and a target probability distribution $P_X$ on a generic $k$-point alphabet $\{1, \ldots,k\}$, we seek a globally optimal quantizer $X=f(W)$ such that $X\sim P_X$ while the minimum mean squared error (MMSE) of estimating $W$ from the quantizer output $X$ is minimized. This formulation frees the design from any prespecified reconstruction points: the quantizer is free to choose its induced conditional means and hence its optimal reconstruction points, subject only to matching the prescribed output distribution. 
This is the key difference from the optimal transport based quantizer design where the reconstruction points are fixed \cite{power_cell_cluster} hence only leading to locally optimal quantizers to the problem considered here; this freedom also makes the problem studied in this work considerably more challenging.
The resulting problem can be viewed as a distortion-optimal way to shape the quantization output distribution, thereby enabling quantizer designs that are simultaneously estimation accurate, entropy or information controllable, and channel or downstream task compatible.
As another perspective, specifying the output distribution may be viewed as a form of regularization of distortion minimization, which enables us to derive the optimal quantizer. Without such a regularization, solely minimizing distortion falls back to the classical quantizer design problem, which remains generally intractable.

\section{Problem statement}
Let $W$ be a continuous random variable with a distribution $P_W$ over $\R$ that can be described by a continuous cumulative distribution function (CDF) $F_W$ or a probability density function (PDF) $p_W$.
The goal is to find a quantizer or simply a function $f$ that maps $W$ to $X=f(W)\in\{1,\ldots,k\}$ such that 
\begin{itemize}[leftmargin=*]
    \item the distribution of $X$ follows any specified distribution $P_X$ over $\{1,\ldots,k\}$, and
    \item the MMSE of estimating $W$ from X is minimized among all quantizers satisfying the above condition.
\end{itemize}
Formally, given $P_W$ and $P_X$, this function optimization problem can be stated as
\begin{align}
    \min_{\substack{f:\,\R \rightarrow \{1,\ldots,k\} \\ \text{s.t. } f(W)\sim P_X}} \text{mmse}(W|f(W)) \label{eq:opt1}
\end{align}
where 
\begin{align}
\text{mmse}(W|f(W)) \deq \min_{g:\,\{1,\ldots,k\}\rightarrow\R}\E[(W-g(f(W)))^2]    
\end{align}
for any $f$.
Given the number of values $k$ that can be taken by the quantization output, restricting the range of the quantizer to be 
$\{1,\ldots,k\}$ is of no loss of generality, as there is a one-to-one mapping between $\{1,\ldots,k\}$ and any finite set of cardinality $k$, and such a mapping does not change the MMSE of estimating $W$ from the quantization output \cite[Lemma~1]{MER2022}.

\section{Derivation of optimal quantizer}
\subsection{Problem conversion}
For any quantizer $f:\R\rightarrow\{1,\ldots,k\}$ that maps $W$ to $X=f(W)$, using basic properties of MMSE and conditional expectation, it can be shown that
\begin{align}
    \text{mmse}(W|X) = \E[W^2]- \E[\E[W|X]^2] .\label{eq:mmse_decomp}
\end{align}
As $\E[W^2]$ is determined by $P_W$, whenever $\E[W^2] < \infty$, the optimization problem in \eqref{eq:opt1} is equivalent to 
\begin{align}
    \max_{\substack{f:\,\R \rightarrow \{1,\ldots,k\} \\ \text{s.t. } f(W)\sim P_X}} \E[\E[W|f(W)]^2] . \label{eq:max_E[m^2]}
\end{align}
Under the assumption that $W$ is a continuous random variable, $F_W$ is invertible and we can define
\begin{align}
    Q_W(u) \deq F_W^{-1}(u) , \quad u\in[0,1] .
\end{align}
With a random variable $U$ uniformly distributed over $[0,1]$, the tuple $(Q_W(U), f(Q_W(U))$ has the same joint distribution as $(W, f(W))$ for any $f:\R\rightarrow\{1,\ldots,k\}$, hence
\begin{align}
 \E[\E[W|f(W)]^2] = \E[\E[Q_W(U)|f(Q_W(U))]^2] .
\end{align}
For notational simplicity, we represent $P_X$ as
\begin{align}
    p_i \deq P_X(i) , \quad i=1,\ldots,k .
\end{align}
We have the following lemma.
\begin{lemma}\label{lm:convert_f_S}
The maximum of \eqref{eq:max_E[m^2]} satisfies
\begin{align}
& \max_{\substack{f:\,\R \rightarrow \{1,\ldots,k\} 
\\ 
\textnormal{s.t. } f(W)\sim P_X}} \E[\E[W|f(W)]^2] = 
\nonumber \\
& 
\max_{\substack{\sS_1,\ldots,\sS_k\subset[0,1] 
\\
\textnormal{s.t. } 
\sS_i\cap \sS_j = \emptyset, \, \cup_{i=1}^k \sS_i = [0,1] \\
P_U(\sS_i)=p_i } }
\sum_{i=1}^k p_i \, \left(\frac{1}{p_i}\int_{\sS_i}Q_W(u) \rd u \right)^2 . \label{eq:max_S}
\end{align}
Moreover, an $f^*$ satisfying the constraint and achieving the maximum on the LHS of \eqref{eq:max_S} can be obtained as
\begin{align}
    f^*(w) = i \quad\textnormal{if $w\in Q_W(\sS^*_i)$} , \quad w\in\R \label{eq:f*_def}
\end{align}
where $\sS^*_1,\ldots,\sS^*_k$ are subsets of $[0,1]$ satisfying the constraints and achieving the maximum on the RHS of \eqref{eq:max_S}, and $Q_W(\sS^*_i) \deq \{Q_W(u): u\in \sS^*_i\}$.

\begin{proof}
For any $f:\R\rightarrow\{1,\ldots,k\}$ satisfying $f(W) \sim P_X$,
\begin{align}
& \E[\E[W|f(W)]^2] \nonumber \\
=& \E[\E[Q_W(U)|f(Q_W(U))]^2]  \\
=& 
\sum_{i=1}^k \PP[f(Q_W(U))=i] \, \E[Q_W(U)|f(Q_W(U))=i]^2 \\
=& 
\sum_{i=1}^k p_i \, \left(\frac{1}{p_i}\int_{\sS_i}Q_W(u) \rd u \right)^2 \label{eq:E_E[W|f(W)]^2}
\end{align}
where
\begin{align}
    \sS_i \deq \{u \in [0,1]: f(Q_W(u))=i\} . \label{eq:S_def}
\end{align}
As $f$ induces an ordered partition $(\sW_1,\ldots,\sW_k)$ of $\R$, which may also be called a labeled partition,
with 
\begin{align}
\sW_i \deq \{w\in\R: f(w)=i\}
\end{align} 
and $(F_W, Q_W)$ is a pair of bijections between $\R$ and $[0,1]$, the $k$-tuple of subsets $(\sS_1,\ldots,\sS_k)$ of $[0,1]$ defined in \eqref{eq:S_def} constitute an ordered partition of $[0,1]$, satisfying
$\sS_i = F_W(\sW_i)$, $\sS_i\cap \sS_j = \emptyset$, $\cup_{i=1}^k \sS_i = [0,1]$ and $P_U(\sS_i)=p_i$.
This proves that
\begin{align}
& \max_{\substack{f:\,\R \rightarrow \{1,\ldots,k\} 
\\ 
\textnormal{s.t. } f(W)\sim P_X}} \E[\E[W|f(W)]^2] 
\le \nonumber \\
& 
\max_{\substack{\sS_1,\ldots,\sS_k\subset[0,1] 
\\
\textnormal{s.t. } \sS_i\cap \sS_j = \emptyset, \, \cup_{i=1}^k \sS_i = [0,1], \\
P_U(\sS_i)=p_i}}
\sum_{i=1}^k p_i \, \left(\frac{1}{p_i}\int_{\sS_i}Q_W(u) \rd u \right)^2 . \label{eq:max_f_ineq}
\end{align}
Next is to show that the equality in \eqref{eq:max_f_ineq} can be achieved by $f^*$ defined in \eqref{eq:f*_def}.
To this end, let $\sS^*_1,\ldots,\sS^*_k$ be subsets of $[0,1]$ satisfying the constraints and achieving the maximum on the RHS of \eqref{eq:max_f_ineq}.
For $f^*$ defined as
\begin{align}
    f^*(w) = i \quad\textnormal{if $w\in Q_W(\sS^*_i)$} , \quad w\in\R , \nonumber
\end{align}
it maps $\R$ to $\{1,\ldots,k\}$ and $\PP[f(W)=i]=\PP[W\in Q_W(\sS^*_i)] = \PP[Q_W(U)\in Q_W(\sS^*_i)] = \PP[U\in\sS^*_i] = p_i$, hence $f^*$ satisfies the constraint on the LHS of \eqref{eq:max_f_ineq}.
Moreover, for this $f^*$,
\begin{align}
   \{u: f^*(Q_W(u))=i\} 
   = \{u: Q_W(u) \in Q_W(\sS^*_i)\} 
  =\sS^*_i .
\end{align}
Together with \eqref{eq:E_E[W|f(W)]^2}, it follows that $f^*$ achieves the RHS of \eqref{eq:max_f_ineq}. This proves the lemma.
\end{proof}

\end{lemma}

Lemma~\ref{lm:convert_f_S} converts the quantizer optimization problem in \eqref{eq:max_E[m^2]}
to an ordered partition optimization problem of the unit interval on the RHS of \eqref{eq:max_S}.
Its proof also shows a one-to-one correspondence between any admissible quantizer $f$ satisfying the constraint on the LHS of \eqref{eq:max_S} and an admissible ordered partition $(\sS_1,\ldots,\sS_k)$ satisfying the constraints on the RHS of \eqref{eq:max_S} through \eqref{eq:S_def} and $\eqref{eq:f*_def}$.

\subsection{Solution to converted problem}
The converted optimization problem through Lemma~\ref{lm:convert_f_S} can be restated as
\begin{align}
\max_{\substack{\sS_1,\ldots,\sS_k\subset[0,1] 
\\
\textnormal{s.t. } \sS_i\cap \sS_j = \emptyset, \, \cup_{i=1}^k \sS_i = [0,1], \\
P_U(\sS_i)=p_i}}
\sum_{i=1}^k p_i \, m_i^2 \label{eq:converted_max_Em^2}
\end{align}
where
\begin{align}
    m_i &\deq \frac{1}{p_i}\int_{\sS_i}Q_W(u) \rd u 
    = \E[Q_W(U)|U\in\sS_i].
\end{align}
For the $f:\R\rightarrow\{1,\ldots,k\}$ induced by an admissible ordered partition $(\sS_1,\ldots,\sS_k)$ of $[0,1]$ through \eqref{eq:f*_def}, 
\begin{align}
    m_i = \E[W|f(W)=i] .
\end{align}
The set of all admissible ordered partitions 
\begin{align}
\{(\sS_1,\ldots,\sS_k): 
& \, \sS_i\subset [0,1],  \sS_i\cap \sS_j = \emptyset,  \nonumber \\
&  \cup_{i=1}^k \sS_i = [0,1], P_U(\sS_i)=p_i\} \label{eq:def_set_ord_part}
\end{align}
can be divided into $k!$ categories according to the order of the values of $m_1, \ldots, m_k \in \R$ induced by each $(\sS_1,\ldots,\sS_k)$.
Specifically, each category can be represented by a permutation $\sigma$ of $\{1,\ldots,k\}$, such that 
\begin{align}
    m_{\sigma(1)} \le \ldots \le m_{\sigma(k)} \label{eq:order_m}
\end{align}
for each $(\sS_1,\ldots,\sS_k)$ in that category.
We have the following lemma.
\begin{lemma}\label{lm:major}
Among the category of the admissible ordered partitions in \eqref{eq:def_set_ord_part} resulting in $m_{\sigma(1)} \le \ldots \le m_{\sigma(k)}$, the one defined by contiguous intervals
\begin{align}
    \sS^*_{\sigma(i)} \deq [q_{i-1}, q_i)  \quad i=1,\ldots,k \label{eq:def_S*_q}
\end{align}
with
\begin{align}
    q_i \deq q_{i-1} + p_{\sigma(i)}
\end{align}
and $q_0 \deq 0$ satisfies
\begin{align}
    \sum_{i=1}^j p_{\sigma(i)} m^*_{\sigma(i)}
    \le
    \sum_{i=1}^j p_{\sigma(i)} m_{\sigma(i)} \quad j=1,\ldots, k-1 \label{eq:major_j}
\end{align}
and 
\begin{align}
    \sum_{i=1}^k p_{\sigma(i)} m^*_{\sigma(i)}
    =
    \sum_{i=1}^k p_{\sigma(i)} m_{\sigma(i)} \label{eq:major_k}
\end{align}
for all other ordered partitions in this category.
\end{lemma}
\begin{proof}
First, by the construction of $\sS^*_{\sigma(i)}$, the order in \eqref{eq:order_m} is met, as $Q_W(u)$ is increasing.
For $j=1,\ldots,k$, we have
\begin{align}
\sum_{i=1}^j p_{\sigma(i)} m^*_{\sigma(i)}
&= \sum_{i=1}^j \int_{\sS^*_{\sigma(i)}} Q_W(u) \rd u \\
&= \int_{\cup_{i=1}^j \sS^*_{\sigma(i)}} Q_W(u) \rd u
\end{align}
and for any other $(\sS_1,\ldots,\sS_k)$ in this category
\begin{align}
\sum_{i=1}^j p_{\sigma(i)} m_{\sigma(i)}
&= \sum_{i=1}^j \int_{\sS_{\sigma(i)}} Q_W(u) \rd u \\
&= \int_{\cup_{i=1}^j \sS_{\sigma(i)}} Q_W(u) \rd u .
\end{align}
Let $\sA_j \deq \cup_{i=1}^j \sS^*_{\sigma(i)}$ and ${\mathsf B}_j \deq \cup_{i=1}^j \sS_{\sigma(i)}$. Then $P_U(\sA_j) = P_U(\mathsf B_j)$, hence
\begin{align}
    P_U(\sA_j\setminus\mathsf B_j) = P_U(\mathsf B_j \setminus \sA_j) . \label{eq:B-A=A-B}
\end{align}
It follows that
\begin{align}
& \sum_{i=1}^j p_{\sigma(i)} m_{\sigma(i)} 
-
\sum_{i=1}^j p_{\sigma(i)} m^*_{\sigma(i)} \nonumber \\
=& 
\int_{\mathsf B_j} Q_W(u) \rd u 
-
\int_{\sA_j} Q_W(u) \rd u \\
=& \int_{\mathsf B_j\setminus\sA_j} Q_W(u) \rd u 
-
\int_{\sA_j\setminus\mathsf B_j} Q_W(u) \rd u \\
\ge& P_U(\mathsf B_j \setminus \sA_j) \Big( \inf_{u\in \mathsf B_j \setminus \sA_j} Q_W(u) - \sup_{u\in \mathsf A_j \setminus \mathsf B_j} Q_W(u)\Big) \label{eq:major_B-A} \\
\ge & 0  \label{eq:major_B>A}
\end{align}
where \eqref{eq:major_B-A} follows from \eqref{eq:B-A=A-B} and the fact that $Q_W(u)$ is increasing in $u$, and \eqref{eq:major_B>A} follows from the fact that $Q_W(u)$ is increasing in $u$ and 
\begin{align}
    \sup( \sA_j\setminus\mathsf B_j )
\le 
    \inf( \mathsf B_j\setminus\sA_j )
\end{align}
by the construction of $\sS^*_{\sigma(i)}$ for $i=1,\ldots,j$ in \eqref{eq:def_S*_q}.
This proves \eqref{eq:major_j}.
The proof of \eqref{eq:major_k} is straightforward as
\begin{align}
    \sum_{i=1}^k p_{\sigma(i)} m^*_{\sigma(i)}
    =
    \sum_{i=1}^k p_{\sigma(i)} m_{\sigma(i)} 
    = \int_{[0,1]} Q_W(u) \rd u  = \E[W]. \nonumber
\end{align}
\end{proof}
Lemma~\ref{lm:major} shows that $(m^*_{\sigma(1)}, \ldots , m^*_{\sigma(k)})$ induced by the ordered partition $(\sS^*_1,\ldots,\sS^*_k)$ defined in \eqref{eq:def_S*_q} weighted-majorizes $( m_{\sigma(1)}, \ldots , m_{\sigma(k)})$ induced by any other ordered partition $(\sS_1,\ldots,\sS_k)$ in the same category.
The following lemma can be viewed as a weighted majorization inequality, a.k.a.\ weighted Karamata inequality \cite{Fuchs_weighted, Majorization_inequalities_IT}. 
We include a proof for completeness.
\begin{lemma}\label{lm:wkata}
For $m^*_1 \le \ldots \le m^*_k \in \R$, $m_{1} \le \ldots \le m_{k} \in \R$, and $p_1,\ldots, p_k \in \R$, if
\begin{align}
    \sum_{i=1}^j p_{i} m^*_{i}
    \le
    \sum_{i=1}^j p_{i} m_{i} \quad j=1,\ldots, k-1 \label{eq:wkata_j}
\end{align}
and 
\begin{align}
    \sum_{i=1}^k p_{i} m^*_{i}
    =
    \sum_{i=1}^k p_{i} m_{i} \label{eq:wkata_k}
\end{align}
then for any convex function $\varphi: \R \rightarrow \R$,
\begin{align}
    \sum_{i=1}^k p_{i} \varphi(m^*_{i})
    \ge
    \sum_{i=1}^k p_{i} \varphi( m_{i} ) .
\end{align}
\end{lemma}
\begin{proof}
For $j=1,\ldots,k$, define
\begin{align}
a_j \deq \sum_{i=1}^j p_i m^*_i, \quad
b_j \deq  \sum_{i=1}^j p_i m_i
\end{align}
and $a_0 = b_0 \deq 0$.
Also, define
\begin{align}
c_i \deq \frac{\varphi(m_i) - \varphi(m^*_i)}{m_i - m^*_i} .
\end{align}
The convexity of $\varphi$ implies that
\begin{align}
    c_i \le c_{i+1} . \label{eq:wkata_c}
\end{align}
We have
\begin{align}
    & \quad \sum_{i=1}^k p_{i} \varphi(m^*_{i})
    -
    \sum_{i=1}^k p_{i} \varphi( m_{i} ) \nonumber \\
    &= \sum_{i=1}^k c_i (a_i - a_{i-1} - b_i + b_{i-1}) \\
    &= \sum_{i=1}^k c_i (a_i - b_i)
    - \sum_{i=1}^k c_i (a_{i-1} - b_{i-1}) \\
    &= c_k(a_k-b_k) + \sum_{i=1}^{k-1}(c_i - c_{i+1})(a_i - b_i) \label{eq:wkata_1}\\ 
    &\ge 0 . \label{eq:wkata_2}
\end{align}
where \eqref{eq:wkata_1} is due to $a_0=b_0=0$, and \eqref{eq:wkata_2} follows from the facts that $a_k-b_k=0$ by \eqref{eq:wkata_k}, $a_i - b_i \le 0$ by \eqref{eq:wkata_j}, and $c_i - c_{i+1} \le 0$ by \eqref{eq:wkata_c}.
\end{proof}
The conventional majorization inequality is a special case of Lemma~\ref{lm:wkata} by taking $p_1=\ldots=p_k=1$ and reversing the indices $(1,\ldots,k)$. 
Together with Lemma~\ref{lm:major}, by taking $\varphi(m)=m^2$, Lemma~\ref{lm:wkata} proves the in-category optimality of the ordered partition $(\sS^*_1,\ldots,\sS^*_k)$ defined in \eqref{eq:def_S*_q}, as stated in the following lemma.
\begin{lemma}\label{lm:in-catagory}
Among the category of the admissible ordered partitions in \eqref{eq:def_set_ord_part} resulting in $m_{\sigma(1)} \le \ldots \le m_{\sigma(k)}$, the one defined by contiguous intervals in \eqref{eq:def_S*_q} achieves the maximum of $\sum_{i=1}^k p_{i} m^{2}_{i}$.
\end{lemma}
With Lemma~\ref{lm:convert_f_S} and Lemma~\ref{lm:in-catagory} in hand, we can state the main result of this work.
\begin{theorem}\label{th:main}
The optimal quantizer that achieves the minimum in \eqref{eq:opt1} can take the form
\begin{align}
X = \sigma^*\big(F_{\sigma^{*-1}(X)}^{-1}(F_W(W))\big) \label{eq:thm_f}
\end{align}
where $\sigma^*$ is the permutation of $\{1,\ldots,k\}$ that maximizes
\begin{align}
    \sum_{i=1}^k \frac{1}{p_i} \Big(\int_{\sS_i}Q_W(u) \rd u \Big)^2 \label{eq:thm_Em^2}
\end{align}
with $\sS_{\sigma(i)} = [q_{i-1}, q_i)$, $q_i = q_{i-1} + p_{\sigma(i)}$, $i=1,\ldots,k$,
and $q_0 = 0$ among all permutations $\sigma$.
\end{theorem}
\begin{proof}
The analysis developed so far shows that the maximization on the RHS of \eqref{eq:max_S} can be achieved by first dividing the set of admissible ordered partitions of $[0,1]$ into $k!$ categories according to the order of $m_1, \ldots, m_k$, and then optimizing within each category.
It follows from Lemma~\ref{lm:in-catagory} that in a category represented by the permutation $\sigma$, the optimal ordered partition takes the form
$\sS_{\sigma(i)} = [q_{i-1}, q_i)$ with $q_i = q_{i-1} + p_{\sigma(i)}$, $i=1,\ldots,k$,
and $q_0 = 0$. 
The globally optimal ordered partition is then determined by the permutation $\sigma^*$ that maximizes \eqref{eq:thm_Em^2}.
It then follows from Lemma~\ref{lm:convert_f_S}
and \eqref{eq:max_E[m^2]} that the quantizer $f^*$ determined by $\sigma^*$ achieves the minimum $\text{mmse}(W|f^*(W))$ while satisfying $f^*(W)\sim P_X$.

It remains to show that this quantizer takes the particular form in \eqref{eq:thm_f}.
To see this, recall that the quantizer determined by $\sigma^*$ is given by \eqref{eq:f*_def} as
\begin{align}
    f^*(w) = \sigma^*(i) \quad\textnormal{if $w\in Q_W(\sS_{\sigma^*(i)})$} , \quad w\in\R . \label{eq:thm_pf_f*}
\end{align}
We have
\begin{align}
w\in Q_W(\sS_{\sigma^*(i)}) 
\Leftrightarrow F_W(w) \in 
\sS_{\sigma^*(i)} ,
\end{align}
$F_W(W)$ has the same distribution as $U$, and the $k$-tuple 
\begin{align}
& \quad 
(P_U(\sS_{\sigma^*(1)}), \ldots, P_U(\sS_{\sigma^*(k)})) 
 \nonumber \\ 
&= (P_{X}(\sigma^*(1)), \ldots, P_{X}(\sigma^*(k))) \\
&= (P_{\sigma^{*-1}(X)}(1), \ldots, P_{\sigma^{*-1}(X)}(k))
\end{align}
is the probability mass function of the random variable $\sigma^{*-1}(X)$, whose generalized inverse CDF maps the contiguous subintervals
$(\sS_{\sigma^*(1)}, \ldots, \sS_{\sigma^*(k)})$ of $[0,1]$ to the indices $(1,\ldots,k)$.
With these observations, we know the function in \eqref{eq:thm_pf_f*} is equivalent to
\begin{align}
f^*(w) = \sigma^*\big(F_{\sigma^{*-1}(X)}^{-1}(F_W(w))\big), \quad w\in\R .
\end{align}
This completes the proof of the theorem.
\end{proof}

\section{Examples and applications}
\subsection{Gaussian input}
When $W$ is Gaussian with mean $\mu$ and standard deviation $s$, 
$   Q_W(u) = \mu + s \Phi^{-1}(u)$
and
$    \int_a^b Q_W(u) \rd u = \mu (b-a) + s(\varphi(\Phi^{-1}(a)) - \varphi(\Phi^{-1}(b))) $
where $\varphi$ and $\Phi$ are the PDF and CDF of the standard Gaussian distribution respectively.
Consider the standard Gaussian case with $P_X = (0.1, 0.2, 0.3, 0.4)$.
Numerical computation of \eqref{eq:thm_Em^2} for $4!=24$ candidate contiguous intervals as quantization regions shows that an optimal permutation $\sigma^*$ is $(1,3,4,2)$, corresponding to the quantization thresholds $(-1.2816, -0.2533, 0.8416)$ and the minimum MMSE $0.1236$.
The optimal quantization regions are shown in Fig.~\ref{fig:gaussian} and the details of $\sigma^{*-1}$ and $P_{\sigma^{*-1}(X)}$ are shown in the table above it. Due to the symmetry of the distribution, $(2,4,3,1)$ is another optimal permutation. It is seen that the quantization regions for larger probability masses are put closer to the mean.

\subsection{Uniform input}
When $W$ is uniformly distributed over an interval $[a,b]$, $ Q_W(u) = (b-a)u+a .   $
In this case, regardless of the output distribution $P_X=(p_1,\ldots,p_k)$, candidate contiguous intervals corresponding to all permutations result in the same MMSE 
\begin{align}
    \frac{(b-a)^2}{12}\sum_{i=1}^k p_i^3 \ge \frac{(b-a)^2}{12k^2} . \label{eq:unif_mmse}
\end{align}
This is because $\int_q^{q+p} Q_W(u)\rd u$ in this case only depends on the length $p$ of the interval but not its location.
If we further optimize over all output distributions over $\{1,\ldots,k\}$ to minimize the MMSE, the RHS of \eqref{eq:unif_mmse} is achieved when $P_X$ is uniform over $\{1,\ldots,k\}$.
This recovers a classic result in quantization theory, stating that each additional bit of a uniform quantizer for a uniform distribution reduces the mean squared quantization error by a factor of $4$, or $6.02$ dB \cite{quant_1998}.
This example also illustrates a new approach to distortion-minimizing quantizer design: for each output distribution, find the optimal quantizer under that constraint, then optimize over the output distribution.

\subsection{Uniform output: entropy and information maximizing quantization}
When the output distribution $P_X$ is uniform over $\{1,\ldots,k\}$, again the permutation does not matter, as contiguous intervals $(\sS_{\sigma(1)}, \ldots, \sS_{\sigma(k)})$ corresponding to all permutations are the same.
Moreover, in this case the output entropy achieves the maximum, $\log k$, and the mutual information between the input $W$ and output $X$ is also maximized, as
\begin{align}
    I(W;X) = H(X) - H(X|W) = H(X)
\end{align}
where the second equality is due to the fact that $X$ is determined by $W$.
In other words, when the output distribution is uniform, the quantizer derived in Theorem~\ref{th:main} minimizes ${\rm mmse}(W|X)$ among all quantizers that maximize $I(W;X)$.
This yields a principled design of quantizers that simultaneously achieves distortion minimization and information maximization.

\subsection{Explicit control of output entropy}
When quantizing an i.i.d.\ source sequence, by selecting a low-entropy \(P_X\), one can intentionally produce a low entropy rate output sequence that can be further compressed, e.g.\ by arithmetic or Huffman coding, trading off a controlled increase in distortion for a lower expected bit rate \cite{Gallager_2008}. 
On the other hand, choosing \(P_X\) to be uniform to maximize the output entropy given $k$ is also viable, as it makes each output bit count without further compression.

\begin{table}[t]
\centering
\begin{tabular}{r|cccc}
\hline
$x$\,\,\, & 1 & 2 & 3 & 4 \\
\hline
$P_X(x)$ & 0.1 & 0.2 & 0.3 & 0.4 \\
$\sigma^*(x)$ & 1 & 3 & 4 & 2 \\
$\sigma^{*-1}(x)$ & 1 & 4 & 2 & 3 \\
$P_{\sigma^{*-1}(X)}(x)$ & 0.1 & 0.3 & 0.4 & 0.2 \\
\hline
\end{tabular}
\end{table}
\begin{figure}[t]
    \centering
\includegraphics[width=0.4\textwidth]{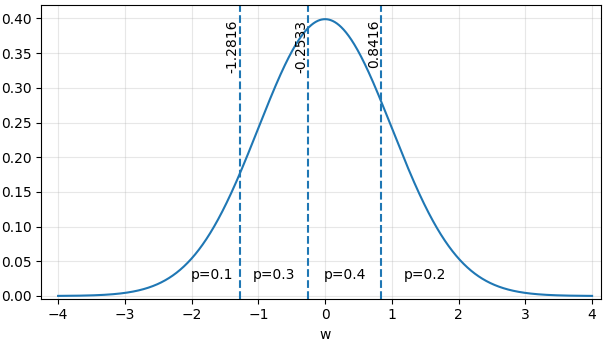}
    \caption{Optimal quantization regions for standard Gaussian with output distribution $(0.1, 0.2, 0.3, 0.4)$.}
    \label{fig:gaussian}
\end{figure}

\subsection{Channel matching}
In communication settings where the quantizer output is sent over a physical channel with stringent input distribution requirements, shaping \(P_X\) provides a mechanism for {channel matching}. 
For example, a channel might be peak power limited, requiring an input distribution that avoids high amplitude symbols, or it might be sequence constrained, e.g.\ those in magnetic recording or optical fibers, where certain symbols must occur with specific frequencies to maintain synchronization or prevent nonlinear effects \cite{MarcusRothSiegel}.
Channel matching ensures that the statistical properties of the quantization output align with the physical or regulatory constraints of the transmission channel.


\subsection{Clustering and data anonymization}
The optimal quantizer derived here may provide more principled solutions to 1D clustering problems \cite{Ckmeans_1ddp,Nielsen_1d_cluster} by assuming a continuous distribution of the data.
Finally, in data anonymization and constrained disclosure settings, output distribution shaping can be used to regulate the statistical property of released symbols while retaining as much estimation accuracy as possible \cite{REBOLLOMONEDERO2013185,dist_prsv_k_anonym}.

\section{Conclusion and possible extensions}
We derived the optimal quantizer of a real-valued random variable $W$ such that the output distribution $P_X$ is as specified, while the MMSE of estimating $W$ from $X$ is minimized.
Important extensions of this work include 
\begin{itemize}[leftmargin=*]
\item Joint source-channel coding

In this scenario, the quantization output $X$ is to be sent over a noisy channel $P_{Y|X}$, resulting in a noise-corrupted source representation $Y$. 
The methods developed in this work can be used to design a quantizer with specified output distribution $P_X$ while minimizing the MMSE of estimating $W$ from the channel output $Y$.
    
\item Communication with feedback 

In a sequential feedback communication scenario, at the $t$th step, the source $W$ together with the previous channel outputs $(Y_1, \ldots, Y_{t-1})$ are encoded to a new channel input $X_t$, and the new channel output $Y_t$ is subsequently fed back to the encoder. Methods developed in this work can be used to design encoders in this feedback communication scenario such that the MMSE of estimating $W$ from $(Y_1, \ldots, Y_t)$ is minimized while the mutual information $I(W; Y_1, \ldots, Y_t)$ is maximized for each $t$ \cite{feedback_maxI_minD_xu}.
\end{itemize}


{
    \small
    \bibliography{main}
}

\end{document}